\documentclass[a4wide]{article}
\usepackage[utf8x]{inputenc}
\pdfoutput=1
\usepackage{graphicx}
\usepackage{amsmath}
\usepackage{amssymb}
\usepackage{amsthm}
\usepackage{subfigure}
\usepackage{titling}
\usepackage[pdftitle={A lower bound for the tree-width of vital linkages in planar graphs - A planar graph for which solutions to a DPP with k terminals use all vertices of a grid of size exponential in k},pdfauthor={Philipp Klaus Krause}]{hyperref}

\newcommand{\coloneqq}{\mathrel{\mathop:}=}

\newcommand{\bigmid}{\;\big|\;}

\newcommand{\inp}{\textbf{Input}}

\newcommand{\quest}{\textbf{Question}}
\newcommand{\rahmen}[1]{\fbox{\parbox{10cm}{
\begin{tabular*}{9.8cm}{ll}#1\end{tabular*}
}}}

\newtheorem{definition}{Definition}
\newtheorem*{conj}{Conjecture}
\newtheorem{theorem}{Theorem}
\newtheorem{corollary}{Corollary}

\author{Isolde Adler, Philipp Klaus Krause}
\title{A lower bound for the tree-width of planar graphs with vital linkages }%\\

\begin{document}

\maketitle

\begin{abstract}
The disjoint paths problem asks, given an graph $G$ and $k + 1$ pairs of
terminals $(s_0,t_0), \ldots ,(s_k,t_k)$, whether there are $k+1$ pairwise
disjoint paths $P_0, \ldots ,P_k$, such that $P_i$ connects $s_i$ to $t_i$.
Robertson and Seymour have proven that the problem can be solved in polynomial
time if $k$ is fixed. Nevertheless, the constants involved are 
huge, and the algorithm is far from implementable. 
The algorithm uses a bound on the
tree-width of graphs with vital linkages, and deletion of irrelevant vertices. We
give single exponential lower bounds both for the tree-width of planar graphs with vital linkages, 
and for the size of the grid necessary for finding irrelevant 
vertices.
\end{abstract}

%%%%%%%%%%%%%%%%%%%%%%%%%%%%%%%%%%%%%%%%%

\section{Introduction}

The \emph{disjoint paths problem} is the following problem.

\smallskip\rahmen{
	$\inp$: Graph $G$, terminals $(s_0,t_0),\ldots, (s_k,t_k)\in V(G)^{2(k+1)}$\\
	$\quest$: Are there $k+1$ pairwise vertex disjoint paths\\ $\quad \quad P_0,\ldots,P_k$ in $G$
		such that $P_i$ has endpoints $s_i$ and $t_i$?
	}	

\smallskip\noindent	
It is a classic problem in algorithmic graph theory and it has many applications, e.\,g.
in routing problems, PCB design and VLSI layout~\cite{AggarwalKW00, LeisersonM85}.
It is NP-hard~\cite{Karp}, and it remains NP-hard on planar graphs~\cite{Lynch}. 
Robertson and Seymour proved that for fixed $k$ it is decidable in polynomial 
time~\cite{Robertson}. More precisely, they showed that it 
can be decided  
in FPT time $f(k)\cdot \left|V(G)\right|^3$, 
where $f$ is a computable function. For planar graphs, Reed et 
al.\ \cite{Reed} gave an algorithm that solves the problem 
in time 
$g(k)\cdot \left|V(G)\right|$ for some computable function $g$. 
Both functions $f$ and $g$ are not really made explicit. They are huge towers of exponentiations, 
and the algorithms are
far from being implementable, even for small $k$. 
Hence an important task is to improve the
algorithms towards a better dependency on the parameter $k$. 

We address this problem, exhibiting
a lower bound on the parameter dependency arising from an essential technique used in both 
algorithms~\cite{Robertson,Reed}.
This technique is finding 
\emph{irrelevant vertices}~\cite{gm-irrelevant}, i.e.\ vertices in
the input graph $G$, such that their deletion does not affect the answer to the problem.
This is closely related to the theorem on \emph{vital linkages}~\cite{gm-linkage}, and it is a main
source of the impractical parameter dependency. 
The irrelevant vertices can be guaranteed if $G$ contains a sufficiently large 
subdivided grid as a subgraph (depending on $k$).
Hence an interesting open question is to determine tight bounds on the size of the
subdivided grid, necessary for $G$ to contain an irrelevant vertex. 
In this paper we
give a lower bound by showing that a  
$(2^k+1)\times (2^k+1)$ grid may not suffice -- even in planar graphs.
Despite recent progress \cite{Shorter}, the quest for good upper bounds is
still open.  

\begin{figure}
\centering
\includegraphics{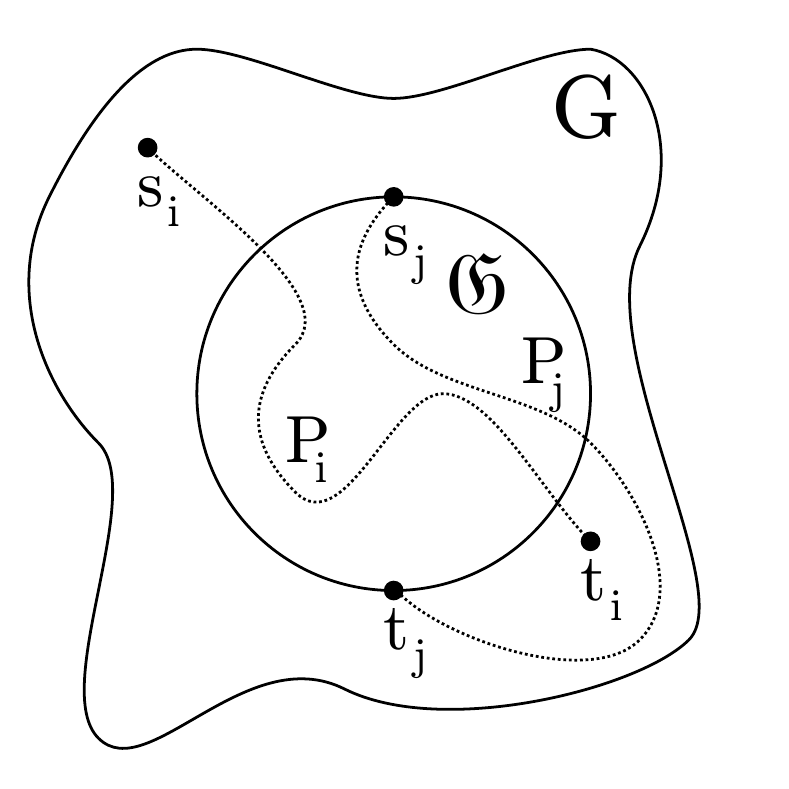}	
\caption{DPP in planar Graph $G$ containing grid $\mathfrak{G}$\label{PGraph}}
\end{figure}

Indeed, grids occur naturally in many graphs that are relevant in practical applications of the
disjoint paths problem, such as PCB and VLSI design. 
For many classes of NP-hard graph problems, attempts have been made to reduce
the computational complexity by restricting the class of input graphs by an
upper bound on a width parameter, such as tree-width and rank-width~\cite{widthparameters}. 
However, grids remain notoriously hard to deal with~\cite{grid-tree, grid-path, grid-rank}.

%%%%%%%%%%%%%%%%%%%%%%%%%%%%%%%%%%%%%%%%%

\section{Preliminaries}

Graphs are finite, undirected and simple. We denote the vertex set of a graph
$G$ by $V(G)$ and the edge set by $E(G)$. Every edge is a two-element subset of
$V(G)$. A graph $H$ is a \emph{subgraph} of a graph $G$, denoted by $H\subseteq G$, if $V(H)\subseteq V(G)$ and
$E(H)\subseteq E(G)$. A \emph{path} in a graph $G$ is a sequence $P=v_1,\ldots,v_n$ of pairwise distinct 
vertices of $G$, such that $\{v_i,v_{i+1}\}\in E(G)$ for all $1\leq i\leq n-1$.
We use standard graph terminology as in~\cite{diestel2005}.

\begin{definition}[Grid]
Let $m,n\geq 1$. The $(m\times
n$) \emph{grid} is the graph $\mathfrak{G}_{m,n}$ given by
\begin{align*}
	V(\mathfrak{G}_{m,n}):=&\big\{1,\ldots,m\big\}\times\big\{1,\ldots,n\big\},\text{ and}\\
	E(\mathfrak{G}_{m,n}):=&\big\{\{(i_1,i_2),(j_1,j_2)\}\bigmid (i_1=j_1\text{ and } |j_2-i_2|=1)\\
&\hspace{2.4cm}\text{ or }(|j_1-i_1|=1\text{ and } i_2=j_2)\big\}
\end{align*}
\end{definition}

A \emph{subdivided grid} is a graph obtained from a grid by replacing some edges of the grid
by pairwise internally vertex disjoint paths of length at least one.
\emph{Embeddings} of graphs in the plane, \emph{planar graphs} and \emph{faces} are defined in the usual way.

\begin{definition}[Inner vertex]
Let $\mathfrak{G}$ be a grid embedded in the plane. An \emph{inner vertex} of $\mathfrak{G}$ 
is a vertex that does not lie on the outer face of $\mathfrak{G}$.
\end{definition}

\begin{definition}[Crossing]
We say that a path \emph{crosses} the grid $\mathfrak{G}$ if it contains an
inner vertex of $\mathfrak{G}$ and its endpoints are not inner vertices of
$\mathfrak{G}$. For $k \in \mathbb{N}$ we say that a path $P = p_0, p_1,
\ldots, p_n$ \emph{crosses $\mathfrak{G}$} $k$ times, if it can be split into
$k$ paths $P_0 = p_0, p_1, \ldots, p_{i_1}, P_1 = p_{i_1}, p_{i_1 + 1}, \ldots
p_{i_2}, \ldots, P_{k - 1} = p_{i_{k - 1}}, p_{i_{k - 1} + 1}, \ldots, p_n$
with each $P_i, i = 0, \ldots k - 1$ crossing $\mathfrak{G}$.
\end{definition}

%%%%%%%%%%%%%%%%%%%%%%%%%%%%%%%%%%%%%%

\section{The lower bound}

From now on we will consider the case of a planar graph $G$ containing a grid
$\mathfrak{G}$ with all $s_i$ and $t_i$ lying on the edge of or outside the
grid $\mathfrak{G}$ (Figure \ref{PGraph}).
Intuitively, we construct our example from a grid $\mathfrak{G}$ of sufficient size. We add endpoints $s_0$ and $t_0$ on the boundary of the grid, mark the areas opposite to the grid as not part of the graph and connect $s_0$ to $t_0$ without crossing the grid. Now we continue to mark vertices by $s_i$ and $t_i$ 
in such a way that $P_i$ has to cross $\mathfrak{G}$ as often as possible (in order to avoid crossing $P_j, j < i$). Once $s_i$ and $t_i$ have been added we remove the area opposite to the grid from $s_i$ from the graph. Figure \ref{ExH} shows the situation after doing this for $i$ up to $2$.
In this construction $P_0$ does not cross the grid at all, while $P_1$ crosses it once and $P_{i+1}$ crosses it twice as often as $P_i$ for $i > 0$: Let $k_i$ be the number of times $P_i$ crosses the grid. $k_0 = 0, k_1 = 1, k_{i + 1} = 2 k_i, k_i = 2^{i - 1}, i > 0$. After the last $P_i$ has been added, the areas opposite to the grid from both $s_i$ and $t_i$ are removed from the graph as seen in Figure \ref{ExLast}.

Formally, to construct problem and graph with $k + 1$ terminals, we use a $(2^k
+ 1) \times (2^k + 1)$ grid. Let the vertices on the left border of the grid be
$n_0, \ldots, n_{2^k}$. Terminals are assigned as follows: $t_0$ is the topmost
vertices on the left border on the grid, $t_1$ the middle vertices on the right border.
For all other terminals: $s_i \coloneqq n_{2^{k-i}}, t_i \coloneqq n_{3 \cdot
2^{k-i}}$. Then add edges going around the $t_i$ to the graph: For $i > 1, t_i
= n_j$ add $n_{j - 1}n_{j + 1}, n_{j - 2}n_{j + 1}, \ldots, n_{j - 2^{k - i} -
1}n_{j + 2^{k - i} - 1}$, and on the right border of $\mathfrak{G}$ do the
analogue for $t_1$. See Figure \ref{ExGrid} for a graph constructed this way.

\begin{figure}
%\centering
\centerline
{
\subfigure[$P_0, \ldots, P_2$\label{ExH}]{\includegraphics[scale=1.4]{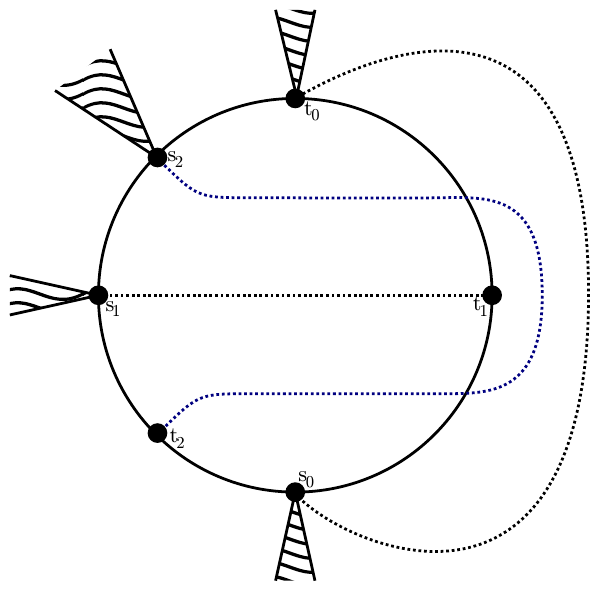}}
\subfigure[$P_0, \ldots, P_3$]{\includegraphics[scale=1.4]{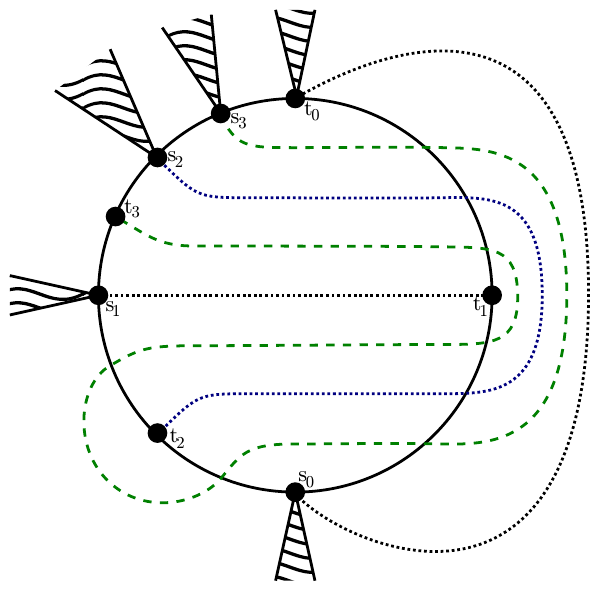}}
}
\centerline
{
\subfigure[$P_0, \ldots, P_4$\label{ExLast}]{\includegraphics[scale=1.4]{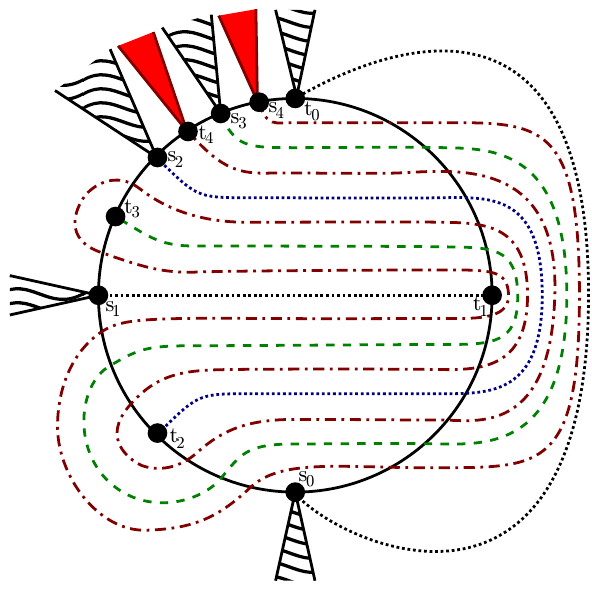}}
\subfigure[Actual graph with grid\label{ExGrid}]{\includegraphics[scale=1.4]{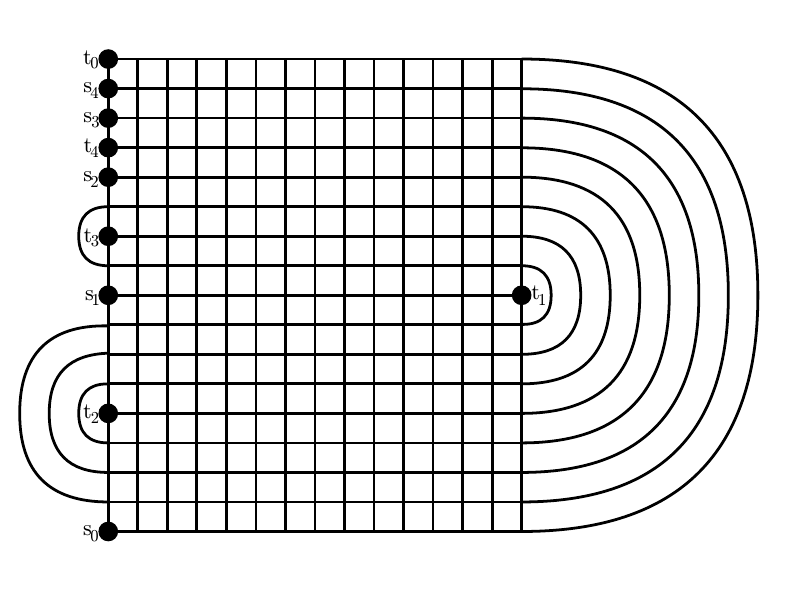}}
}
\caption{Construction of graph and solution}
\end{figure}

\begin{theorem}
There is only one solution to the constructed DPP, all vertices of the graph lie on paths of the solution and the grid is crossed $2^k -1$ times by such paths.
\end{theorem}

\begin{proof}
To connect $s_k$ to $t_k$ we need to cross $\mathfrak{G}$ at least once (Figure \ref{Proof}). However, a solution in which $P_K$ crosses $\mathfrak{G}$ only once would block $P_{k - 1}$. To connect $s_{k - 1}$ to $t_{k - 1}$ $P_k$ has to be routed around $t_{k - 1}$, which requires leaving and reentering $\mathfrak{G}$ (Figure \ref{ProofI}). Thus inductively constructing a solution each $P_i, i > 0$ requires a crossing of $\mathfrak{G}$ and doubles the number of crossings in each $P_j, j > i$. The solution uses all edges on the left side of $\mathfrak{G}$ and uses all but one of the edges on the right side. Thus the only way to connect $s_0$ to $t_0$ is without crossing the grid.
\end{proof}

In particular, $\mathfrak{G}$ has no \emph{irrelevant} vertex in the sense of \cite{gm-irrelevant}.

\begin{corollary}
There is a planar graph $G$ with $k+1$ pairs of terminals such that 
\begin{itemize}
	\item $G$ contains a $(2^k+1)\times (2^k+1)$ grid as a subgraph, 
	\item the disjoint paths problem on this input has a unique solution,
	\item the solution uses all vertices of $G$; in particular, no vertex of $G$ is irrelevant.
\end{itemize}
\end{corollary}

\begin{conj} There is a function $f\in 2^{\mathcal O(k)}$ such that
for every planar input $G$ together with $k+1$ pairs of terminals: if $G$ contains a 
subdivided $f(k)\times f(k)$ grid as a subgraph, then $G$ contains an irrelevant vertex.
\end{conj}

%%%%%%%%%%%%%%%%%%%%%%%%%%%%%%%%%%%%%%

\section{Vital linkages and tree-width}

We refer the reader to~\cite{Bodlaender93} for the definitions of \emph{tree-width} and \emph{path-width}.

\begin{definition}[Vital linkage]
	Let $L$ be a subgraph of $G$ such that every component of $L$ is a path and all vertex of $G$ are vertices of $L$. The \emph{pattern} of $L$ is the set of vertices of degree $1$ in $L$. $L$ is a \emph{vital linkage} in $G$ if there is no other such $L$ that has the same pattern.
\end{definition}

\begin{theorem}[Robertson and Seymour \cite{gm-linkage}]\label{thm:vital}
There are functions $f$ and $g$ such that if $G$ has a vital linkage with $k$ components then $G$ has tree-width at most $f(k)$ and path-width at most $g(k)$.
\end{theorem}

Recall that the $n\times n$ grid has path-width $n$ and tree-width $n$. 
Our example yields a lower bound for $f$ and $g$:

\begin{corollary} Let $f$ and $g$ be as in Theorem~\ref{thm:vital}. Then
$2^{k - 1} + 1 \leq f(k)$ and $2^{k - 1} + 1 \leq g(k)$. 
\end{corollary}

\begin{proof}
Looking at the graph $G$ and DPP constructed above the solution to the DPP is, due to its uniqueness, a vital linkage for the graph $G$. $G$ contains a $(2^k+1) \times (2^k+1)$ grid as a minor. The tree-width of such a grid is $2^k+1$ , its path-width $2^k+1$ \cite{grid-path}. Thus we get lower bounds  $2^{k - 1} + 1 \leq f(k), g(k)$ for the functions $f$ and $g$.
\end{proof}

\begin{figure}
\centerline
{
\subfigure[$P_0, P_4$\label{Proof}]{\includegraphics[scale=1.4]{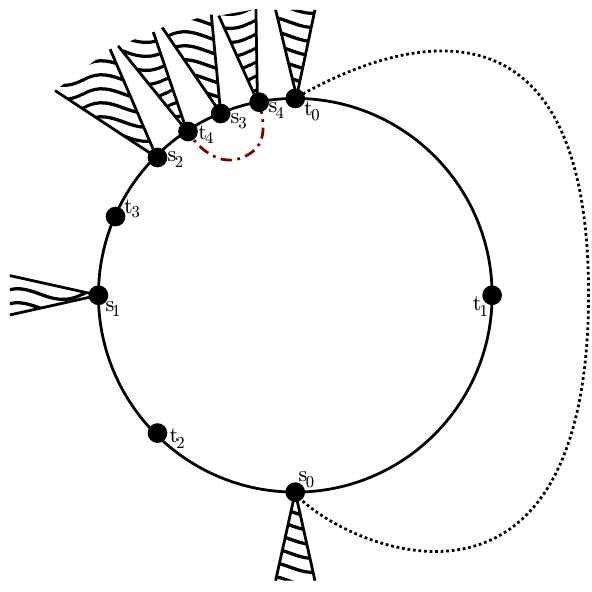}}
\subfigure[$P_0, P_4, \ldots, P_3$\label{ProofI}]{\includegraphics[scale=1.4]{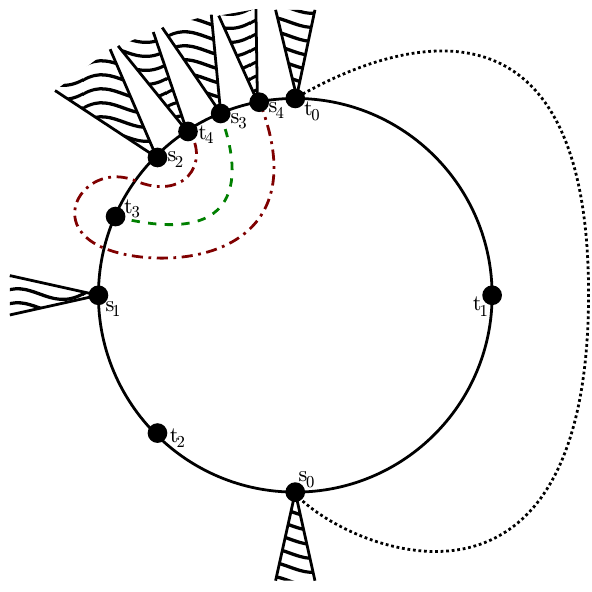}}
}
\centerline
{
\subfigure[$P_0, P_4, \ldots, P_2$]{\includegraphics[scale=1.4]{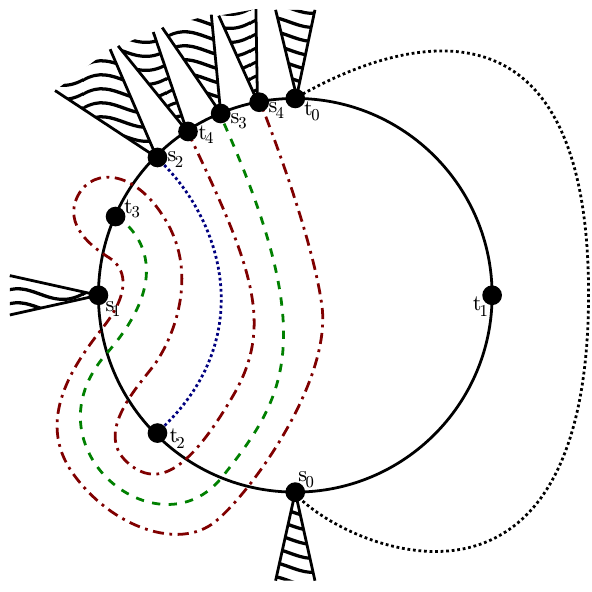}}
\subfigure[$P_0, P_4, \ldots, P_1$]{\includegraphics[scale=1.4]{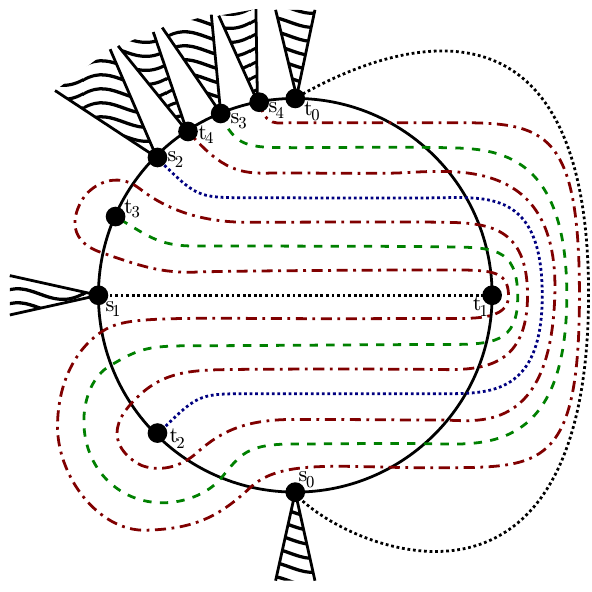}}
}
\caption{Optimality of solution}
\end{figure}
%%%%%%%%%%%%%%%%%%%%%%%%%%%%%%%%%%%%%%

\newpage

\bibliographystyle{plain}
\bibliography{bib}

\end{document}